\documentclass[sigconf,nonacm]{style/acmart}

\usepackage[]{svg}
\usepackage{multirow}
\usepackage{caption}
\usepackage{subcaption}
\usepackage{algorithm}
\usepackage{algpseudocode}
\usepackage{afterpage}
\usepackage{tablefootnote}
\usepackage{appendix}
\usepackage{float}
\usepackage{gensymb}

\usepackage{amsthm}
\newtheorem{definition}{Definition}

\usepackage{xcolor}

\definecolor{codegreen}{rgb}{0,0.6,0}
\definecolor{codegray}{rgb}{0.5,0.5,0.5}
\definecolor{codepurple}{rgb}{0.58,0,0.82}
\definecolor{backcolour}{rgb}{0.95,0.95,0.92}
\usepackage{listings}
\lstdefinestyle{mystyle}{
    backgroundcolor=\color{backcolour},   
    commentstyle=\color{codegreen},
    keywordstyle=\color{magenta},
    numberstyle=\tiny\color{codegray},
    stringstyle=\color{codepurple},
    basicstyle=\ttfamily\footnotesize,
    breakatwhitespace=false,         
    breaklines=true,                 
    captionpos=b,                    
    keepspaces=true,                 
    numbers=left,                    
    numbersep=5pt,                  
    showspaces=false,                
    showstringspaces=false,
    showtabs=false,                  
    tabsize=2
}
\newbool{showcomments}

   \booltrue{showcomments}

    \newcommand{\pinaforecomment}[4]{%
    \ifbool{showcomments}{%
        \colorbox{#1}{\textcolor{#4}{\parbox{.8\linewidth}{#2: #3}}}%
    }{}%
}

\lstdefinelanguage{json}{
    basicstyle=\ttfamily\small,
    showstringspaces=false,
    breaklines=true,
    literate=
     *{:}{{{\color{red}:}}}{1}
      {,}{{{\color{blue},}}}{1}
      {?}{{{\color{red}? }}}{1}
      {\{}{{{\color{blue}\{}}}{1}
      {\}}{{{\color{blue}\}}}}{1}
      {[}{{{\color{blue}[}}}{1}
      {]}{{{\color{blue}]}}}{1}
}
\usepackage[most]{tcolorbox}
\tcbuselibrary{skins,breakable,listings}
\usepackage{siunitx}
\tcbuselibrary{listings,breakable,skins}
\newtcblisting{queryblock}[2][]{%
  breakable,
  enhanced,
  colback=black!2,
  colframe=black!35,
  arc=2pt,
  boxrule=0.5pt,
  left=6pt,right=6pt,top=4pt,bottom=4pt,
  title={#2},
  fonttitle=\bfseries,
  listing only,
  listing options={
    language=json,
    basicstyle=\ttfamily\footnotesize,
    columns=fullflexible,
    keepspaces=true,
    breaklines=true,
    breakatwhitespace=true,
    showstringspaces=false,
  },
  #1
}

\begin{document}


\graphicspath{ {figures/}{auto_commit_fig/}{auto_fig/} }

\newcommand{\latexfile}[1]{\input{sections/#1}}

\newcommand{\avikomment}[1]{\pinaforecomment{green}{Avik}{#1}{black}}
\newcommand{\nrscomment}[1]{\pinaforecomment{violet}{Nicole}{#1}{white}}
\newcommand{\lukascomment}[1]{\pinaforecomment{blue}{Lukas}{#1}{white}}


\title{3D Spatial Pattern Matching}

\author{Nicole R. Schneider}
\email{nsch@umd.edu}
\affiliation{%
  \institution{University of Maryland}
  \city{College Park}
  \country{USA}
}

\author{Avik Das}
\email{adas1236@terpmail.umd.edu}
\affiliation{%
  \institution{University of Maryland}
  \city{College Park}
  \country{USA}
}

\author{Lukas Arzoumanidis}
\email{lukas.arzumanidis@hcu-hamburg.de}
\orcid{https://orcid.org/0000-0001-6668-1695}
\affiliation{%
  \institution{HafenCity University}
  \city{Hamburg}
  \country{Germany}
}

\author{Abhijeet Ghodgaonkar}
\email{aghodgao@terpmail.umd.edu}
\affiliation{%
  \institution{University of Maryland}
  \city{College Park}
  \country{USA}
}

\author{Hanan Samet}
\email{hjs@cs.umd.edu}
\affiliation{%
  \institution{University of Maryland}
  \city{College Park}
  \country{USA}
}

\author{Youness Dehbi}
\email{youness.dehbi@hcu-hamburg.de}
\orcid{https://orcid.org/0000-0003-0133-4099}
\affiliation{%
  \institution{HafenCity University}
  \city{Hamburg}
  \country{Germany}
}

\begin{abstract}
Spatial pattern matching is the process of matching query entities and constraints with database entities and relations.
It has many applications, including similar region search, housing market search, landmark search, and road network matching.
To our knowledge, all existing spatial pattern matching approaches frame the problem in a 2 dimensional space, where entities lie in a cartesian plane and relationships defined between them are contained in 2 dimensions.
However, this problem framing has significant limitations when searching for real world entities that have height in addition to position.
To address this limitation, we extend spatial pattern matching to 3 dimensions and provide a generalized definition of the problem.
We describe a subgraph matching algorithm capable of resolving 3D spatial patterns over distance relations and release two 3D spatial pattern matching datasets, one synthetic and one containing real 3D building data from the city of Hamburg, Germany.
We test our subgraph matching algorithm on both datasets and present results as a baseline for future methods to build upon.

\end{abstract}


\maketitle

\section{Introduction}


Spatial pattern matching is the process of matching query entities and constraints with database entities and relations.
It has many applications, including similar region search, housing market search, landmark search, and road network matching~\cite{Schneider2024}.
To our knowledge, all existing spatial pattern matching approaches frame the problem in a 2 dimensional space, where entities lie in a cartesian plane and relationships defined between them are similarly contained in 2 dimensions.
Semantically enriched building features such as doors and windows provide useful spatial context that cannot be captured using 2D spatial pattern matching.

\begin{figure}[htbp]
  \centering
  \includegraphics[width=0.9\linewidth]{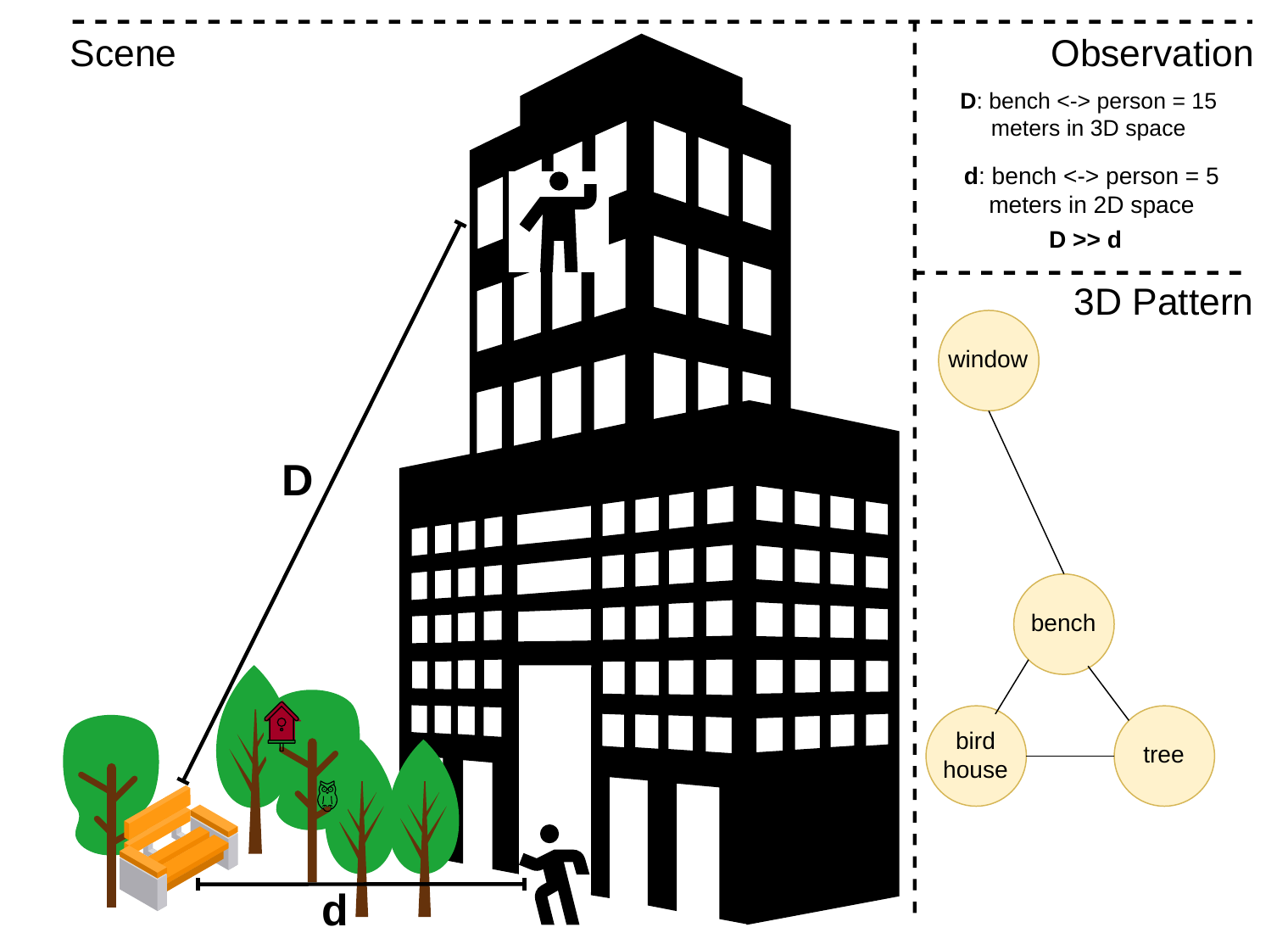}
  \caption{A person looking out window or from ground level may appear co-located in 2D, yet have substantially different distances to other objects in a 3D scene. Since vertical distance influences spatial relationships, we generalize spatial pattern matching into 3 dimensions to support richer search.}
\Description{3D spatial scene}
  \label{fig:pinn}
\end{figure}

However, the 2-dimensional framing of the problem has significant limitations when searching for real world entities that have height in addition to position, because they exist in 3-dimensional space.
To illustrate the limitation of 2 dimensional spatial pattern matching, consider the landmark search use case for the problem~\cite{Osul2023,Osul2023b}, where a user may recall spatial features at different heights, especially in a dense urban area with tall buildings.
For example, in Figure \ref{fig:pinn}, a bird-watcher looking out a window might see about 15 meters away a bench, tree, and birdhouse which hangs from the tree.
With 2D spatial pattern matching, the query for such a scene would be limited to 2 dimensional distances between the building (ground level) to the objects in the scene, incorrectly capturing the user's distance estimate, and failing to capture the role of height in the spatial scene.

\begin{queryblock}{2D Query Constraints}
(bench, building, ?)
(birdhouse, building, ?)
(tree, building, ?)
(bench, birdhouse, 1 to 3 meters)
(birdhouse, tree, 1 to 3 meters)
(bench, tree, 1 to 3 meters)
\end{queryblock}

Here, the first relations involving the user's viewpoint from the building's window, cannot be defined in 2 dimensional space, so they get compressed to ground level.
However, the user does not know the approximate distance between the bench and the building at ground level, since their perspective is from a tall window.
Even if the user could estimate the distance in the 2D projection of the scene, the scene itself forms a non-planar graph that is not satisfiable using 2D dimensional spatial pattern search.
The 3D spatial pattern version of the same query matches the scene accurately, capturing the height of the building implicitly as part of the window distance constraints in the query.

\begin{queryblock}{3D Query Constraints}
(bench, window, 14 to 16 meters)
(birdhouse, window, 14 to 16 meters)
(tree, window, 14 to 16 meters)
(bench, birdhouse, 1 to 3 meters)
(birdhouse, tree, 1 to 3 meters)
(bench, tree, 1 to 3 meters)
\end{queryblock}


To address the present 2D limitation of spatial pattern matching, we extend the problem to 3 dimensions and provide a subgraph matching algorithm~\footnote{\href{https://github.com/adas1236/SpatialPatternMatching3d}{https://github.com/adas1236/SpatialPatternMatching3d}} capable of resolving 3D spatial patterns over distance relations.
To aid in further research, we also release two 3D spatial pattern matching datasets, one synthetic and one containing real 3D building data from the city of Hamburg in Germany.
We test our subgraph matching algorithm on both datasets and present results as a baseline for future methods to build upon.
\section{Problem Definition}

Metric spatial pattern matching is a search over metric relations that characterize how far apart entities are in space~\cite{Schneider2024b}. 
We define metric relations between entities as ranges of distances between entities, each defined by a distance interval and a sign \{exclusion in, exclusion out, mutual exclusion and multiple inclusion\}~\cite{Fang2019}.

\begin{definition}[Spatial Query Pattern]
A spatial query pattern is a graph
$P = P(V, E)$ 
of spatial objects $V$ and edges $E$.
For each $v_i \in V$, there is an associated keyword $w_i$ and
for each edge $e_{i, j} = (v_i, v_j) \in E$, there is an interval $[\ell_{i,j}, u_{i, j}]$ and a sign $\sigma \in \{\leftarrow, \rightarrow, \leftrightarrow, -\}$. 

Let $o_s$ and $o_t$ be any two objects matched with $(v_i, v_j)$, i.e., they have keywords $w_i$ and
$w_j$ respectively and satisfy the spatial constraints. Then, following the convention in \citeauthor{Fang2019}~\cite{Fang2019}, the signs mean:
\begin{itemize}
    \item $v_i \rightarrow v_j$ $[v_i$ excludes $v_j]$: No object with keyword $w_j$ in $D$ should have a distance less than $\ell_{i,j}$ from $o_s$.
    \item $v_i \leftarrow v_j$ $[v_j$ excludes $v_i]$: No object with keyword $w_i$ in $D$ should have a distance less than $\ell_{i,j}$ from $o_t$.
    \item $v_i \leftrightarrow v_j$ $[$mutual exclusion$]$: No object with keyword $w_j$ in $D$ should have a distance less than $\ell_{i,j}$ from $o_s$, and the distance of any object with keyword $w_i$ in $D$ should be at least $\ell_{i,j}$ from $o_t$.
    \item $v_i - v_j$ $[$mutual inclusion$]$: The occurrence of any object pairs (other than $o_s$ and $o_t$) with keywords $w_i$ and $w_j$ in $D$ with distance shorter than $\ell_{i,j}$ is allowed.
\end{itemize}

\end{definition}

\begin{definition}[Spatial Database]
For our purposes, a spatial database is a weighted graph $D = D(V, E)$
of spatial objects $V$ and edges $E$ such that for each object $v_i \in V$, there is an associated set of keywords $W_i$.
\end{definition}

\begin{definition}[$k$-Dimensional Metric Spatial Database]
To say a spatial database $D$ is $k$-dimensional is to assert there is an isometric embedding $D \hookrightarrow \mathbb{R}^k$.
\end{definition}

\section{Related Work}

\subsection{2D Metric Spatial Pattern Matching}
\textit{Spacekey}~\cite{Fang2018,Fang2018b,Fang2019} proposes \textit{Multi-Pair Join} and \textit{Multi-Star Join} which use subgraph matching for spatial pattern search over metric relations.
Their Multi-Pair-Join algorithm has a worst-case complexity $\mathcal{O}(|P|\zeta |D|^2+\xi)$ where $\zeta$ is a sampling threshold in range $[0,1]$ and $\xi$ is the maximal number of partial matches. 
Their Multi-Star-Join algorithm has a worst-case complexity of $\mathcal{O}(|D|^4+|P||D|^2+\xi)$ time, but in practice is faster than the Multi-Pair-Join because it uses additional pruning criteria to eliminate partial matches during the join process~\cite{Fang2019}.


\textit{ESPM: Efficient Spatial Pattern Matching}~\cite{Chen2019} also performs spatial pattern matching over metric relations, but extends the work of \citeauthor{Fang2018}~\cite{Fang2019} by adding a step that uses a set of Inverted Linear Quadtrees, one per entity keyword, to prune unpromising edges before running Multi-Star-Join. 
By eliminating unpromising nodes early and carefully constructing the join order, ESPM scales better than Multi-Star-Join, despite having a similar theoretical worst-case complexity that is still exponential in the number of query entities.


Other metric spatial pattern matching approaches use example pattern queries that require the user identify existing patterns from the database to use for search~\cite{Chen2022,Luo2017,Zhang2022}.
This form of the problem is well suited for similar region search, but is ill-suited to other use cases, like location recall and landmark search, where the user vaguely recalls a place and hopes to find it again. These methods are also 2-dimensional in nature, since the examples are selected from a 2D map.
Mapping the example-based spatial pattern matching paradigm into 3 dimensions using our problem definition is an interesting area of future work.

\subsection{3D Spatial Data Structures}
Many spatial data structures have been designed for 3D data in metric space, as described by \citeauthor{Samet2006}~\cite{Samet2006}.
Examples include the Octree~\cite{Meagher1982}, the Vantage-Point Tree~\cite{Yianilos1993}, and the M-Tree~\cite{Ciaccia1997} which are designed to support proximity queries for nearest neighbor search in generalized metric spaces.
We leverage the Octree in our implementation of the 3D spatial pattern matching algorithm.

\subsection{3D City Modeling}
CityGML is an Open Geospatial Consortium (OGC) standard~\footnote{\url{https://www.ogc.org/standards/citygml/}} for the representation, storage, and exchange of semantic 3D city models and constitutes one of the fundamental data models for urban digital twins. 
In contrast to conventional 3D city models, which primarily focus on geometric or graphical representations, CityGML integrates semantic, topological, geometric, and appearance-related properties of urban objects within a unified data model.

Enrichment and reconstruction of semantic 3D city models from sparse and dense observations is an active research area. 
Examples include the estimation of building storeys from single street-view images~\citep{hcu_1248} and street-level facade openings from dense point clouds~\cite{Tang_2025_CVPR}. 
Other works, such as 
Nguyen et al.~\cite{nguyen_phdthesis_2024} detect spatio-semantic changes between large CityGML models by representing each model as a graph and comparing the properties and geometry of corresponding nodes. 
This addresses matching in 3D city models represented as semantic graphs rather than as purely geometric objects and explicitly incorporates spatial proximity through Euclidean distance when matching points.
\section{Method}

\subsection{Insight}
In the simplest case, using a naive brute force subgraph matching algorithm, the $k$-D spatial pattern matching problem can be solved by trivially extending the distance formula 
$d = \sqrt{(x_2 - x_1)^2 + (y_2 - y_1)^2}$ to $k$ dimensions:
\begin{equation} \label{3d-distance}
    d = \Big( \textstyle\sum_{i=1}^{k} (x^{(i)}_2 - x^{(i)}_1)^2 \Big)^{1/2}
\end{equation}

However, this is impractical, since the problem is NP-hard~\cite{Fang2019}, and all recent approaches in the 2D case use spatial indexes and careful join order to efficiently solve the problem.
We introduce 2 key insights that allow us to adapt the existing 2D metric spatial pattern matching algorithms~\cite{Chen2019} to work in arbitrarily many dimensions.

\begin{enumerate}
    \item Extending the minimum bounding rectangle to a minimum bounding rectangular prism.
    \item Replacing the Inverted Linear Quadtree with an Inverted Hyperoctree.
\end{enumerate}

With these adjustments, spatial indexing can be used to efficiently prune candidates and determine matches in $k$-dimensions.

\begin{algorithm}
\caption{$k$-dimensional SPM}
\label{alg:3Despm}
\begin{algorithmic}[1]
\Require $IHO$, $P$ \Comment{Inverted Hyperoctree $IHO$, Pattern $P$}
\For{$l \gets 1$ \textbf{to} $L$}
    \State Derive the order $O_l$ of computing $n$-matches;
    \For{each edge $e$ in $O_l$}
        \State Prune unpromising nodes;
        \State Compute the $n$-matches for $e$;
    \EndFor
\EndFor
\State Derive the order $O$ of computing $e$-matches;
\State Identify skip-edges by using $O$;
\For{each edge $e$ in $O$}
    \If{$e$ is not a skip edge}
        \State Prune unpromising objects;
        \State Compute the $e$-matches for $e$;
    \EndIf
\EndFor
\State Derive the order $\Gamma$ of joining $e$-matches;
\State $\Psi \gets$ join $e$-matches following the order $\Gamma$;
\State \Return $\Psi$ \Comment{all the matches of $P$}
\end{algorithmic}
\end{algorithm}

\subsection{$k$-D Spatial Pattern Matching Algorithm}

Following ESPM~\cite{Chen2019}, we first compute n-matches for edges level by level, then compute e-matches for edges, and finally join e-matches to compute all the matches of the spatial pattern.
Our $k$-dimensional Spatial Pattern Matching method (Algorithm~\ref{alg:3Despm}) takes as input a spatial pattern P and the Inverted Hyperoctree $IHO = \{IHO_i\}$, where $IHO_i$ is the Octree of keyword $w_i$ for all objects in a set $D$.

\begin{definition}[$n$-match]
Let $n_i$ and $n_j$ be two nodes in $IHO_i$ and $IHO_j$ respectively,
and $b_i$ and $b_j$ be their Minimum Bounding Rectangular Prisms (MBRPs).
The node pair $(n_i, n_j)$ is an \emph{$n$-match} for the edge
$(v_i, v_j)$ if
\begin{itemize}
  \item \textbf{Case} $v_i - v_j$:
        $d_{\min}(b_i, b_j) \leq u_{ij}$ and
        $d_{\max}(b_i, b_j) \geq l_{ij}$, where
        $d_{\min}(\cdot)$ and $d_{\max}(\cdot)$ are the minimum and
        maximum distances using equation~\ref{3d-distance} between two MBRPs respectively.
  \item \textbf{Case} $v_i \rightarrow v_j$:
        There is no node
        $n'_j \in IO_j\ (\neq n_j)$ such that
        $d_{\max}(b_i, b'_j) < l_{ij}$.
\end{itemize}
\end{definition}

\begin{definition}[$e$-match]
A pair of objects $o_s, o_t \in D$ is an \emph{$e$-match} for the
edge $e = (v_i, v_j) \in E$, if $o_s$ and $o_t$ have keywords $w_i$ and
$w_j$ respectively, and they satisfy the spatial constraints of $e$.
\end{definition}

\begin{lemma}
Let $D$ be a set of spatial objects and let $P = P(V, E)$ be a spatial pattern. For each keyword $w_i$, let $D_i$ be the set of all objects with keyword $w_i$. Let $|D_w| = \max\{|D_i|: \text{$w_i$ is a keyword}\}$. Then $k$-dimensional SPM takes $O\left(k 4^k |D_w|^2 + k|D_w|^n\right)$ time, where $n = |V|$ and $m = |E|$.

\end{lemma}

\begin{proof}
In $k$-dimensional space, the inverted quadtree of ESPM~\cite{Chen2019} is generalized to an inverted hyperoctree, where each node has $2^k$ children. Thus, each surviving node pair generates at most $2^k \cdot 2^k = 4^k$ child-node pairs. This means the $n$-match computation takes time $O\left(k m 4^k \big(\frac{|D_w|}{s}\big)^2\right)$, the $e$-match computation takes time $O(km |D_w|^2)$, and the joining of $e$-matches takes time $O(k|D_w|^n)$, where the estimates are obtained from \citet{Chen2019}, and the linear factor of $k$ is due to the linear time to compute distance or check geometric constraints in $k$-dimensional space. As in \citet{Chen2019}, we typically have $n, m \ll |D_w|$ and $n \geq 2$. Thus, the total computation time is $O\left(k 4^k |D_w|^2 + k|D_w|^n\right)$.
\end{proof}
\section{Experiments}
\subsection{Datasets}
We perform experiments on both synthetic and real-world 3D data.

\paragraph{Synthetic Dataset}
We devise 4 synthetic datasets consisting of 10,000, 100,000, 1 Million, and 10 Million objects, respectively, with each object being a point in a cubic domain with one to three keyword class labels. 
For each dataset we vary sparsity by adjusting the domain size and the keyword frequency, giving 12 configurations in total.
We devise 20 unique query pattern structures and aggregate results over them, measuring total time and peak RAM.

\paragraph{Hamburg Dataset}
To evaluate our algorithm in a real urban environment, the underlying city model must provide a semantically rich representation of buildings, streets, and urban furniture. 
Representations derived solely from the extrusion of 2D footprints, such as those in OpenStreetMap, are inadequate because they lack detailed semantic and geometric information, like facade structures, roof characteristics, openings, materials, and other object-level attributes that distinguish between otherwise similar geometries. 
Our Hamburg dataset is based on the Level of Detail 2 (LOD2) semantic 3D city model of Hamburg, Germany~\footnote{\url{https://geoportal-hamburg.de/?map=3D}}, provided in the CityGML format~\cite{groeger_citygml2_2012}. 
The original CityGML data were translated into a property knowledge graph and subsequently loaded into a Neo4j graph database, enabling knowledge extraction and graph manipulation using the 3DCityKG framework~\footnote{\url{https://github.com/tum-gis/3dcitykg}}.
The city model was provided by the State Office for Geoinformation and Surveying of Hamburg. 
It includes information on roof type, roof height, building height, building use, the number of storeys, and the 3D location of all modeled building parts.
We enrich the knowledge graph with additional building attributes, including roof material information following the methodology proposed by \citet{lukas_roofs}\footnote{\url{https://zenodo.org/records/18896614}}. 
3D data on trees was also provided by the Geoinformation and Surveying of Hamburg~\footnote{\url{https://metaver.de/trefferanzeige?docuuid=24513F73-D928-450C-A334-E30037945729&q=Baeume\%20Hamburg}}.
We also add approximately 20 individual street-level facade openings, including windows and doors.
We devise 5 unique query pattern structures and aggregate results over them, measuring total time and peak RAM.


\subsection{Algorithm and Hardware Settings}
For the Octree, we set the minimum depth to 1, the maximum depth to 12, and the splitting threshold to 64.
The algorithm is implemented in Python, without using libraries.
We compute the first 1000 matches for each query.
All experiments are run using 32 CPUs and 249GB RAM.

\subsection{Results}

We present results on both the real and synthetic datasets in Table~\ref{tab:espm3d_scalability}.
Scalability tests on synthetic data up to 10 million objects show good performance for sparse and medium sparsity datasets (total ~20 to an hour minutes to evaluate all queries for each setting).
On dense datasets, we run into memory issues. This is because the dense settings leads to less selective pruning, resulting in far more node pairs to analyze for $e$-matches. 
On the real Hamburg data, we complete all queries in about 4 minutes with reasonable memory utilization.
Since we propose the first 3D spatial pattern matching method, our results focus on scalability testing rather than comparison to any known baseline.
We find that the algorithms we adapt are efficient for sparse and medium density graphs, but have memory limitations on dense graphs.
We suggest memory efficient 3D spatial pattern matching as a useful line of future research.

\section{Conclusion}

This paper extends the spatial pattern matching problem to $k$-dimensions and proposes a subgraph matching algorithm for arbitrary dimensions capable of resolving spatial patterns over distance relations.
We demonstrate the scalability of our algorithm on two datasets that we construct and open source, including both synthetic data and real building height data from the city of Hamburg, Germany.

\begingroup

\begin{table}[htbp]
\centering
\setlength{\tabcolsep}{2pt}
\renewcommand{\arraystretch}{0.8}
\scriptsize
\begin{tabular}{lrrrrrr}
\toprule
\multicolumn{7}{c}{\textbf{Synthetic}} \\
\hline
& \multicolumn{2}{c}{\textbf{Sparse}} & \multicolumn{2}{c}{\textbf{Medium}} & \multicolumn{2}{c}{\textbf{Dense}} \\
\cmidrule(lr){2-3}\cmidrule(lr){4-5}\cmidrule(lr){6-7}
\textbf{$|D|$} & \textbf{Time (s)} & \textbf{Peak RAM (GB)} & \textbf{Time (s)} & \textbf{Peak RAM (GB)} & \textbf{Time (s)} & \textbf{Peak RAM (GB)} \\
\midrule
10K & 0.72 & 0.0 & 0.39 & 0.0 & 0.31 & 0.0  \\
100K & 18.45 & 0.1 & 10.24 & 0.1 & 23.38 & 0.3  \\
1M & 115.96 & 0.6 & 106.42 & 0.6 & OOM & OOM  \\
10M & 1232.15 & 5.7 & 3786.64 & 217.8 & OOM & OOM \\
\midrule
\multicolumn{7}{c}{\textbf{Hamburg}} \\
\hline
\textbf{$|D|$} & \multicolumn{3}{c}{\textbf{Time (s)}} & \multicolumn{3}{c}{\textbf{Peak RAM (GB)}} \\
\midrule
628,477 & \multicolumn{3}{c}{246.83} & \multicolumn{3}{c}{6.6} \\
\bottomrule
\end{tabular}
\caption{Scalability results. Out of memory marked OOM.}
\label{tab:espm3d_scalability}
\end{table}

\endgroup



\bibliographystyle{style/ACM-Reference-Format.bst} 
\bibliography{main.bib} \label{bibliography}

@inproceedings{Fang2018,
  title={SpaceKey: exploring patterns in spatial databases},
  author={Fang, Yixiang and Cheng, Reynold and Wang, Jikun and Budiman, Lukito and Cong, Gao and Mamoulis, Nikos},
  booktitle={2018 IEEE 34th International Conference on Data Engineering (ICDE)},
  pages={1577--1580},
  year={2018},
  organization={IEEE},
}

@inproceedings{Fang2018b,
  title={On spatial pattern matching},
  author={Fang, Yixiang and Cheng, Reynold and Cong, Gao and Mamoulis, Nikos and Li, Yun},
  booktitle={2018 IEEE 34th International Conference on Data Engineering (ICDE)},
  pages={293--304},
  year={2018},
  organization={IEEE}
}

@book{groeger_citygml2_2012,
	author = {Gröger, Gerhard and Kolbe, Thomas H. and Nagel, Claus and Häfele, Karl-Heinz},
	title = {{OGC City Geography Markup Language (CityGML) Encoding Standard}},
	year = {2012},
	month = {4},
	publisher = {Open Geospatial Consortium (OGC)},
	numpages = {344},
	note = {OGC 12-019, Version 2.0.0, International Standard}
}

@InProceedings{Tang_2025_CVPR,
    author    = {Tang, Wenzhao and Li, Weihang and Liang, Xiucheng and Wysocki, Olaf and Biljecki, Filip and Holst, Christoph and Jutzi, Boris},
    title     = {Texture2LoD3: Enabling LoD3 Building Reconstruction With Panoramic Images},
    booktitle = {Proceedings of the IEEE/CVF Conference on Computer Vision and Pattern Recognition (CVPR) Workshops},
    month     = {June},
    year      = {2025},
    pages     = {2041-2051}
}

@OTHER{
	hcu_1248,
	author = {Lukas Arzoumanidis and Al Maimun As Samee and Elmehdi Kanna and Son Nguyen and Youness Dehbi},
	title = {Domain-Adaptive Object Detection for Enriching Semantic 3D City Models with Building Storeys from Street-View Images},
	year = {2026},
        series = {ISPRS Annals of the Photogrammetry, Remote Sensing and Spatial Information Sciences},
	booktitle = {ISPRS Congress 2026 Toronto, Canada},
	url = {https://repos.hcu-hamburg.de/handle/hcu/1248},
}

@article{Fang2019,
  title={Evaluating pattern matching queries for spatial databases},
  author={Fang, Yixiang and Li, Yun and Cheng, Reynold and Mamoulis, Nikos and Cong, Gao},
  journal={The VLDB Journal},
  volume={28},
  pages={649--673},
  year={2019},
  publisher={Springer},
}

@phdthesis{nguyen_phdthesis_2024,
	author = {Nguyen, Huynh Duc An Son},
	title = {{Automatic Detection and Interpretation of Changes in Massive Semantic 3D City Models}},
	year = {2024},
	school = {Technical University of Munich},
	url = {https://mediatum.ub.tum.de/1748695},
}

@Article{lukas_roofs,
AUTHOR = {Arzoumanidis, Lukas and Nguyen, Son H. and Johannsen, Lara and Rothaut, Filip and Li, Weilian and Dehbi, Youness},
TITLE = {Object Detection for the Enrichment of Semantic 3D City Models with Roofing Materials},
JOURNAL = {ISPRS Annals of the Photogrammetry, Remote Sensing and Spatial Information Sciences},
VOLUME = {X-4/W6-2025},
YEAR = {2025},
PAGES = {9--16},
DOI = {10.5194/isprs-annals-X-4-W6-2025-9-2025}
}

@article{Chen2019,
  title={ESPM: Efficient spatial pattern matching},
  author={Chen, Hongmei and Fang, Yixiang and Zhang, Ying and Zhang, Wenjie and Wang, Lizhen},
  journal={IEEE Transactions on Knowledge and Data Engineering},
  volume={32},
  number={6},
  pages={1227--1233},
  year={2019},
  publisher={IEEE},
}

@inproceedings{Schneider2024b,
author = {Schneider, Nicole R. and O'Sullivan, Kent and Samet, Hanan},
title = {Graph-based Spatial Pattern Matching: A Theoretical Comparison},
year = {2024},
isbn = {9798400711077},
publisher = {Association for Computing Machinery},
address = {New York, NY, USA},
url = {https://doi.org/10.1145/3678717.3691227},
doi = {10.1145/3678717.3691227},
booktitle = {Proceedings of the 32nd ACM International Conference on Advances in Geographic Information Systems},
pages = {505–508},
numpages = {4},
location = {Atlanta, GA, USA},
series = {SIGSPATIAL '24}
}

@inproceedings{Schneider2024,
  TITLE = {{The Future of Graph-based Spatial Pattern Matching (Vision Paper)}},
  AUTHOR = {Schneider, Nicole R. and O'Sullivan, Kent and Samet, Hanan},
  BOOKTITLE = {{40th IEEE International Conference on Data Engineering, ICDE 2024 -- SEAGraph Workshop}},
  publisher="IEEE",
  organization="IEEE",
  pages={360--364},
  ADDRESS = {Utrecht, Netherlands},
  YEAR = {2024},
  MONTH = May,
}

@inproceedings{Osul2023b,
  author  = {O'Sullivan, Kent and Schneider, Nicole R. and Samet, Hanan},
title = {COMPASS: Cardinal Orientation Manipulation and
Pattern-Aware Spatial Search},
year = {2023},
publisher = {Association for Computing Machinery},
booktitle = {Proceedings of the 2nd ACM SIGSPATIAL International Workshop on Searching and Mining Large Collections of Geospatial Data},
location = {Hamburg, Germany},
series = {GeoSearch'23}
}

@inproceedings{Osul2023,
  author  = {O'Sullivan, Kent and Schneider, Nicole R. and Rasheed, Aleeza and Samet, Hanan},
title = {GESTALT: Geospatially Enhanced Search with Terrain
Augmented Location Targeting},
year = {2023},
publisher = {Association for Computing Machinery},
booktitle = {Proceedings of the 2nd ACM SIGSPATIAL International Workshop on Searching and Mining Large Collections of Geospatial Data},
location = {Hamburg, Germany},
series = {GeoSearch'23}
}

@book{Samet2006,
  title={Foundations of multidimensional and metric data structures},
  author={Samet, Hanan},
  year={2006},
  publisher={Morgan Kaufmann}
}

@inproceedings{Yianilos1993,
author = {Yianilos, Peter N.},
title = {Data structures and algorithms for nearest neighbor search in general metric spaces},
year = {1993},
publisher = {Society for Industrial and Applied Mathematics},
address = {USA},
booktitle = {Proceedings of the Fourth Annual ACM-SIAM Symposium on Discrete Algorithms},
pages = {311–321},
numpages = {11},
location = {Austin, Texas, USA},
series = {SODA '93}
}

@inproceedings{Ciaccia1997,
  title={M-tree: An E cient access method for similarity search in metric spaces},
  author={Ciaccia, Paolo and Patella, Marco and Zezula, Pavel},
  booktitle={Proceedings of the 23rd VLDB conference, Athens, Greece},
  pages={426--435},
  year={1997}
}

@article{Meagher1982,
  title={Geometric modeling using octree encoding},
  author={Meagher, Donald},
  journal={Computer graphics and image processing},
  volume={19},
  number={2},
  pages={129--147},
  year={1982},
  publisher={Elsevier}
}

@article{Chen2022,
  title={Example-based spatial pattern matching},
  author={Chen, Yue and Feng, Kaiyu and Cong, Gao and Kiah, Han Mao},
  journal={Proceedings of the VLDB Endowment},
  volume={15},
  number={11},
  pages={2572--2584},
  year={2022},
  publisher={VLDB Endowment},
}

@inproceedings{Luo2017,
  title={Seq: Example-based query for spatial objects},
  author={Luo, Siqiang and Hu, Jiafeng and Cheng, Reynold and Yan, Jing and Kao, Ben},
  booktitle={Proceedings of the 2017 ACM on Conference on Information and Knowledge Management},
  pages={2179--2182},
  year={2017}
}

@INPROCEEDINGS{Zhang2022,
  author={Zhang, Hanyuan and Luo, Siqiang and Shi, Jieming and Nathan Yan, Jing and Sun, Weiwei},
  booktitle={2022 IEEE 38th International Conference on Data Engineering (ICDE)}, 
  title={Example-based Spatial Search at Scale}, 
  year={2022},
  volume={},
  number={},
  pages={539-551},
  doi={10.1109/ICDE53745.2022.00045}
}

\end{document}